\documentclass[12pt, a4paper]{article}
 \usepackage[T1]{fontenc}        
	 \usepackage{amsmath,amsthm}     
	 \usepackage{amsfonts, amssymb}  
 \usepackage{float}              
 \usepackage{url}                
\usepackage{dsfont}		
\usepackage{graphicx}     
\usepackage{subfig}
\usepackage{xcolor} 
\usepackage{mathdots} 
\usepackage{enumerate}
\usepackage{enumitem}
\usepackage{array}
\usepackage[urlcolor=black,linkcolor=black,citecolor=black,colorlinks=true]{hyperref}
\usepackage[left=2.4cm, right=2.4cm, top=2.25cm, bottom=3.25cm]{geometry}

\usepackage{comment}
\usepackage{pict2e}
\usepackage[utf8]{inputenc} 
\usepackage{placeins}

\DeclareMathOperator*{\argmax}{arg\,max}

\newtheorem{theoreme}{Th\'{e}or\`{e}me}[section]
\newtheorem{lemme}[theoreme]{Lemma}

\newtheorem{proposition}[theoreme]{Proposition}

\newtheorem{remarque}[theoreme]{Remark}
\newtheorem{hypothese}{Assumption}[section]

\title{Modeling frequency distribution above a priority in presence of IBNR}

\author{Nicolas Baradel\footnote{Inria, CMAP, CNRS, \'{E}cole polytechnique, Institut Polytechnique de Paris, 91200 Palaiseau, {nicolas.baradel@polytechnique.edu}.}}

\begin{document}

\maketitle

\begin{abstract}

In reinsurance, Poisson and Negative binomial distributions are employed for modeling frequency. However, the incomplete data regarding reported incurred claims above a priority level presents challenges in estimation. This paper focuses on frequency estimation using Schnieper's framework \cite{schnieper1991separating} for claim numbering. We demonstrate that Schnieper's model is consistent with a Poisson distribution for the total number of claims above a priority at each year of development, providing a robust basis for parameter estimation. Additionally, we explain how to build an alternative assumption based on a Negative binomial distribution, which yields similar results. The study includes a bootstrap procedure to manage uncertainty in parameter estimation and a case study comparing assumptions and evaluating the impact of the bootstrap approach.

\end{abstract}
\section{Introduction}

In his profession, a reinsurer has to quote prices for excess of loss covers. Generally, the reinsurer estimates the frequency and severity distributions. For the frequency, the most common choices are the Poisson and Negative binomial distributions.  The Poisson distribution can be viewed as the natural distribution in an ideal world: when all claims are independent and occur with a non-random intensity, the distribution is Poisson. However, when there is some uncaptured randomness, the variance is greater than the mean. In the particular case where the intensity of a Poisson distribution follows a Gamma distribution, the overall distribution is known to be a Negative binomial one.

\medbreak

The data the reinsurer receives are often incomplete: only the reported incurred claims above a certain threshold, typically known as the priority. In the context of excess of loss, \cite[Schnieper]{schnieper1991separating} proposed a model that separates the IBNR (Incurred But Not Reported) into what he termed \textit{true IBNR}: newly reported claims, and \textit{IBNER (Incurred But Not Enough Reported)}: variation in estimated cost over time. 

\medbreak

Although \cite[Mack]{mack1993distribution} used some of the ideas from the Schnieper's method, it has not received much attention. Major contributions based on the Schnieper model include \cite{liu2009predictive} and \cite{liu2009bootstrap}. In the former, the author derives an estimator for the mean square error of the reserves. In the latter, they proposed a non-parametric bootstrap procedure to estimate the distribution of the reserve. Additionally, some valuable insights inspired by Schnieper’s model are found in \cite{ohlssonclaims} and \cite{ohlssonclaims2}. In the former, the author uses Schnieper’s approach to determine the implicit part of the IBNER in the Chain Ladder reserve. In the latter, they adapt the methodology in order to also separate the paid from the incurred claims.

\medbreak

Schnieper also addressed a special case: claim numbering above a priority. The frequency of claims exceeding the priority over time is divided into: claims that newly reach the priority and claims that fall below it. In particular, he proposed assuming a Poisson distribution for claims that reach the priority and a Binomial distribution for claims that drop below it.

\medbreak

In this paper, we focus on frequency estimation in presence of incomplete data, specifically the reported incurred above a priority, using Schnieper' framework for claim numbering. We show that the total number of claims in his model follows a Poisson distribution at each year of development. Consequently, this framework is consistent with a Poisson distribution for the total number of claims above a priority and provides a consistent framework for parameter estimation. We also propose an alternative assumption based on a Negative binomial distribution, which yields similar results. We show that the total number of claims also follows a Negative binomial distribution for each year of development and we provide an estimation procedure. Additionally, we address claim reserving by providing the distribution of ultimate claim numbers, conditioned on current incurred claims.

\medbreak

The paper is organized as follows. Section 2 presents Schnieper's general model, with a review of key estimators. Section 3 covers claim numbers above a priority. The first part deals with the Schnieper assumption, from which we derive additional results. Specifically, we obtain the distribution of the total number of claims for both purposes: quotation and reserving. The second part presents an alternative assumption under which we show that the total claim number above a priority follows a Negative binomial distribution. Section 4 describes a bootstrap procedure for each case, addressing uncertainty in parameter estimation. Finally, Section 5 provides a case study comparing assumptions and evaluating the impact of the bootstrap approach and its contribution to the different assumptions.

\section{The general model}\label{general}

The Schnieper model, with an aim of excess of loss cover, separates two different behaviors in the IBNR data:

\begin{itemize}
    \item The occurrence of newly reported claims, which are assumed to arise randomly based on the level of exposure ;
    \item The progression of previously reported claims, which is determined by the current known amounts.
\end{itemize}
The Schnieper's framework requires more summary statistics than the aggregated evolution of the incurred claims, which we introduce below, where $n \geq 1$ denotes the number of years, $1 \leq i \leq n$ represents the occurrence year, and $1 \leq j \leq n$ represents the development year:

\begin{itemize}
    \item The random variables $(N_{i,j})_{1 \leq i, j \leq n}$ represent the total amount of new excess claims, referring to claims that have not been recorded as excess claims in previous development years ;
    \item The random variables $(D_{i,j})_{1 \leq i, j \leq n}$ represent the decrease in the total claims amount between development years $j-1$ and $j$, concerning claims that were already known in development year $j-1$.
\end{itemize}

The $(D_{i,j})_{1 \leq i, j \leq n}$ can be negative in the event of an increase and, by construction, $D_{i, 1} = 0$ for all $1 \leq i \leq n$.

\bigskip

Given $(N_{i,j})_{1 \leq i, j \leq n}$ and $(D_{i,j})_{1 \leq i, j \leq n}$, the cumulative incurred data $(C_{i,j})_{1 \leq i, j \leq n}$ can be calculated using the following iterative process:

\begin{equation}\label{XND}
    \begin{aligned}
    &C_{i, 1} = N_{i, 1}, & & 1 \leq i \leq n \\
    &C_{i,j+1} = C_{i,j} + N_{i,j+1} - D_{i,j+1}, & & 1 \leq i \leq n, \ 1 \leq j \leq n-1
    \end{aligned}
\end{equation}
We also introduce non-negative exposures $(E_i)_{1 \leq i \leq n}$ that are assumed to be known and associated with the data mentioned above.
Finally, we introduce the following filtration:
    \[
        \mathcal{F}_{k} := \sigma(N_{i,j}, D_{i,j} \mid i + j \leq k + 1), \ \ k \geq 1.
    \]
The current available information is $\mathcal{F}_{n}$. In the context of Schnieper's general model, the following assumption is made:

    \begin{hypothese}
	\leavevmode
        \begin{itemize}
            \item[H1] The random variables $(N_{i_{1}, j}, D_{i_{1}, j})_{1 \leq j \leq n}$ and $(N_{i_{2}, j}, D_{i_{2}, j})_{1 \leq j \leq n}$ are independent for $i_{1} \not= i_{2}$. 
            \item[H2] For $1 \leq j \leq n$, there exists $\lambda_{j} \geq 0$ and for $1 \leq j \leq n - 1$, there exists $\delta_{j} \leq 1$ such that 
                \[
                \begin{aligned}
                    \mathbb{E}(N_{i, j} \mid \mathcal{F}_{i + j - 2}) &= \lambda_{j}E_{i}, \ \ &1 \leq i \leq n,\\
                    \mathbb{E}(D_{i, j+1} \mid \mathcal{F}_{i + j - 1}) &= \delta_{j}C_{i, j}, \ \ &1 \leq i \leq n.
                \end{aligned}
                \]
            \item[H3] For $1 \leq j \leq n - 1$, there exist $\sigma_{j}^{2} \geq 0$ and $\tau_{j}^{2} \geq 0$ such that
                \[
                \begin{aligned}
                    Var(N_{i,j} \mid \mathcal{F}_{i + j - 2}) &= \sigma_{j}^{2}E_{i}, \ \ &1 \leq i \leq n, \\
                    Var(D_{i, j+1} \mid \mathcal{F}_{i + j - 1}) &= \tau_{j}^{2}C_{i,j}, \ \ &1 \leq i \leq n.
                \end{aligned}
                \]
        \end{itemize}
    \end{hypothese}

The evolution of the $(D_{i,j})_{1 \leq i, j \leq n}$ follows the same process as in Mack’s model \cite{mack1993distribution}, using incurred claims as the exposure combined with a development factor. However, the new claims generated by the $(N_{i,j})_{1 \leq i, j \leq n}$ represent an additional additive component that depends on the exposure.

From the above assumption, Schnieper introduced the following estimators for the $\lambda$'s and the $\delta$'s.

            \begin{equation}\label{main_estimators}                   
            \begin{aligned}
            \widehat{\lambda}_{j} &:= \frac{\sum_{i = 1}^{n - j + 1}N_{i, j}}{\sum_{i = 1}^{n - j + 1}E_{i}}, \\
            \widehat{\delta}_{j} &:= \frac{\sum_{i = 1}^{n - j}D_{i, j+1}}{\sum_{i = 1}^{n - j}C_{i,j}}.
            \end{aligned}
                \end{equation}

They are obviously biasfree estimates of the $\lambda$'s and $\delta$'s respectively. Additionally, they are the best linear estimators in $\frac{N_{i,j}}{E_{i}}$ and $\frac{D_{i,j+1}}{C_{i,j}}$ respectively as a consequence of \textit{H3}.

Schnieper developed a way to estimate the expected value of $C_{n+1, n}$ and the pure premium for the following year, including the total reserve along with its associated mean square error based on estimators of the $\sigma$'s and $\tau$'s that he provides. We now present the model for claim number, which can be considered a specific instance of the broader model.

\section{The model for claim numbers}

Schnieper also dealt with a special case: the number of claims above a priority, which is the focus of this paper. From this point onward, for a given priority level, we define

    \begin{itemize}
        \item[\textbullet] The random variables $(N_{i, j})_{1 \leq i, j \leq n}$ represent the number of new excess claims pertaining to accident year $i$ in development year $j$ (were below the priority or not reported in development year $j-1$) ;
        \item[\textbullet] The random variables $(D_{i, j})_{1 \leq i, j \leq n}$ represent the number of claims that exceeded the priority in development year $j-1$ but have since decreased in cost to fall below the priority in development year $j$.
    \end{itemize}

In this context, the $(D_{i,j})$ are now non-negative and remain bounded by $(C_{i, j-1})$. Additionally, all data are integer. 

\bigskip

Schnieper proposed that the new claims, denoted as $(N_{i,j})$, follow a Poisson distribution, while the claims decreasing below the priority, represented by $(D_{i,j})$, follow a Binomial distribution. The reasoning is as follows: claims that rise above the priority (either new or previously below the threshold) occur independently at non random intensity, aligning with a Poisson distribution. Meanwhile, each claim above the priority has a certain probability of falling below the threshold each year, occurring independently and leading to a binomial distribution. The next subsection will provide a detailed explanation of these assumptions and additional results as stated in \cite{schnieper1991separating}, including the finding that $(C_{i,j})$ follows a Poisson distribution at each date.

\subsection{The Poisson assumption}

The following assumption corresponds to \cite[Assumptions $A_1'' - A_2''$]{schnieper1991separating}. We add \emph{H1'} to slighty reinforce \emph{H1}.

    \begin{hypothese}
	\leavevmode
        \begin{itemize}
            \item[H1'] The random variables $(N_{i_{1}, j}, D_{i_{1}, j})_{1 \leq j \leq n}$ and $(N_{i_{2}, j}, D_{i_{2}, j})_{1 \leq j \leq n}$ are independent for $i_{1} \not= i_{2}$, and the pairs $(N_{i, j}, D_{i, j})$ are independent for all $1 \leq i, j \leq n$.
            \item[H2'] For $1 \leq j \leq n$, there exists $\lambda_{j} \geq 0$ and for $1 \leq j \leq n - 1$, there exists $0 \leq \delta_{j} \leq 1$ such that
                \[
                \begin{aligned}
		N_{i,j} \mid \mathcal{F}_{i + j - 2} &\sim \mathcal{P}\left(\lambda_{j} E_{i}\right), \quad &1 \leq i \leq n,\\
		D_{i,j+1} \mid \mathcal{F}_{i + j - 1} &\sim \mathcal{B}(C_{i, j}, \delta_{j}), \quad &1 \leq i \leq n.\\
                \end{aligned}
	\]
        \end{itemize}
    \end{hypothese}

Note that \emph{H1'} implies \emph{H1}, and \emph{H2'} implies \emph{H2} and \emph{H3} from Section \ref{general}. \emph{H1'} adds independance between $N_{i,j}$ and $D_{i,j}$ on the same dates, while \emph{H2'} specifies the distributions. We assume now that the assumptions \textit{H1'} and \emph{H2'} hold. Given these, \cite{schnieper1991separating} showed that the $\widehat{\lambda}$'s and $\widehat{\delta}$'s defined in (\ref{main_estimators}) are also the maximum likelihood estimators and are efficient, with their variances given by the inverse of the Fisher information.

\medbreak

From what Schnieper stated in this framework, we shall show additional results. Specifically, under \emph{H1}' and \emph{H2'}, the $(C_{i,j})_{1 \leq i, j \leq n}$ follow a Poisson distribution. Before stating the result, recall first a classic lemma that will be essential.
    \begin{lemme}\label{poisson_binomiale}
        Let $N$ be a random variable with distribution $\mathcal{P}(\lambda)$ with $\lambda > 0$ and $D$ a random variable such that $D \mid N \sim \mathcal{B}(N, p)$ for $0 < p < 1$. Then
        
        \[
            N - D \sim \mathcal{P}(\lambda(1-p)).
        \]
    \end{lemme}

        \begin{proof}
    For completeness, we provide the proof of this elementary result . Let $k \in \mathbb{N}$.
    
    \begin{align*}
        \mathbb{P}(N - D = k) &= \sum_{n \in \mathbb{N}}\mathbb{P}\left(\left\{D = n - k\right\} \mid \left\{N = n\right\}\right)\mathbb{P}\left(N = n\right) \\
        &= \sum_{n \geq k}\frac{n!}{k!(n-k)!}p^{n-k}(1-p)^{k}\frac{e^{-\lambda}\lambda^{n}}{n!} \\
        &= \frac{e^{-\lambda}\left[\lambda(1-p)\right]^{k}}{k!}\sum_{n \geq 0}\frac{(p\lambda)^{n}}{n!} \\
        &= \frac{e^{-\lambda(1-p)}\left[\lambda(1-p)\right]^{k}}{k!}.
    \end{align*}
    \end{proof}

             \begin{proposition}\label{loi_poisson}
        For all $1 \leq i,j \leq n$, we have
        \[
            C_{i,j} \sim \mathcal{P}\left(\lambda_{j}'E_i\right),
        \]
        with
        \[
            \lambda_{j}' := \sum_{k = 1}^{j}\lambda_{k}\left(\prod_{\ell = k}^{j-1}(1-\delta_{\ell})\right).
        \]
    \end{proposition}
    
    \begin{proof}
    We prove the lemma by induction. Let $1 \leq i \leq n$ be fixed. By construction $C_{i,1} = N_{i,1} \sim \mathcal{P}(\lambda_{1}E_{i})$.
Assume as the induction hypothesis that $C_{i,j}$ follows a Poisson distribution with parameter $\lambda_{j}'E_i$. Recall the relation in (\ref{XND}): 
    \[
            C_{i,j+1} = C_{i,j} + N_{i,j+1} - D_{i,j+1}.
    \]
    Under \emph{H2'}, by Lemma \ref{poisson_binomiale},
    \[
         C_{i,j} - D_{i,j+1} \sim \mathcal{P}\left(\lambda_{j}'E_i(1-\delta_{j})\right),
    \]
    and finally, under \emph{H1'}:
    \[
        C_{i,j+1} \sim \mathcal{P}\left(\lambda_{j+1}E_{i} + \lambda_{j}'(1-\delta_{j})E_{i}\right) = \mathcal{P}\left(\lambda_{j+1}'E_i\right).
    \]
    \end{proof}

The above result provides the distribution of $C_{n+1, n}$ given the corresponding exposure. In practice, we are also interested in the conditional distribution $C_{i,j} \mid \mathcal{F}_n$ for $i+j > n+1$. Before presenting this result, we introduce an essential lemma.

\begin{lemme}\label{binomiale_binomiale}
    Let $N \in \mathbb{N}^{*}$. Let $D_{1} \sim \mathcal{B}(N, p_{1})$, and $D_{n+1} \mid \{\sum_{i=1}^{n}D_{i} = d\} \sim \mathcal{B}(N - d, p_{n+1})$ for $n \geq 1$. Then
        \[
            N - \sum_{i = 1}^{n}D_{i} \sim \mathcal{B}\left(N, \prod_{i = 1}^{n}(1-p_{i})\right).
        \]
\end{lemme}

\begin{proof}
    We prove the lemma by induction. It is clear that $N - D_{1} \sim \mathcal{B}(N, 1-p_{1})$. Let $n \geq 1$, and assume by induction that
        \[
            N - \sum_{i = 1}^{n} D_{i} \sim \mathcal{B}\left(N, \prod_{i = 1}^{n}(1-p_{i})\right).
        \]
        Let $k \in \{0, \ldots, N\}$. For ease of notation, we introduce $1-\pi_{n} := \prod_{i = 1}^{n}(1-p_{i})$. It follows that
        \begin{align*}
            \mathbb{P}\left[N - \sum_{i=1}^{n+1}D_{i} = k\right] &= \sum_{d = 0}^{N}\mathbb{P}\left[N - \sum_{i=1}^{n+1}D_{i} = k \ \middle| \ \sum_{i=1}^{n}D_{i} = d\right]\mathbb{P}\left(\sum_{i=1}^{n}D_{i} = d\right) \\
            &= \sum_{d = 0}^{N}\mathbb{P}\left[D_{n+1} = N - d - k \ \middle| \ \sum_{i=1}^{n}D_{i} = d\right]\mathbb{P}\left(N - \sum_{i=1}^{n}D_{i} = N - d\right) \\
            &= \sum_{d = 0}^{N-k}\frac{(N-d)!}{k!(N-d-k)!}p_{n+1}^{N-d-k}(1-p_{n+1})^{k}\frac{N!}{d!(N-d)!}\pi_{n}^{d}(1-\pi_{n})^{N-d}
            \\
            &= \frac{N!}{k!(N-k)!}(1-\pi_{n+1})^{k}\sum_{d = 0}^{N-k}\frac{(N-k)!}{(N-d-k)!d!}p_{n+1}^{N-d-k}\pi_{n}^{d}(1-\pi_{n})^{N-d - k}            \\
            &= \frac{N!}{k!(N-k)!}(1-\pi_{n+1})^{k}\left(\pi_{n} + p_{n+1}(1-\pi_{n})\right)^{N-k}
            \\
            &= \frac{N!}{k!(N-k)!}(1-\pi_{n+1})^{k}\pi_{n+1}^{N-k}.
        \end{align*}
\end{proof}

    \begin{proposition}\label{loi_fn}
        For all $1 \leq i,j \leq n$ such that $i+j > n + 1$, we have
        \[
            C_{i,j} \mid \mathcal{F}_{n} \sim \mathcal{B}(C_{i,n-i+1}, \delta_{i,j}') + \mathcal{P}\left(\lambda_{i,j}'E_i\right),
        \]
        where the right-hand side should be interpreted as the sum of two independent random variables, and with
        \begin{align*}
            \delta_{i,j}' &:= \prod_{k = n-i+1}^{j-1}(1-\delta_{k}), \\
            \lambda_{i,j}' &:= \sum_{k = n-i+2}^{j}\lambda_{k}\left(\prod_{\ell = k}^{j-1}(1-\delta_{\ell})\right).
        \end{align*}
    \end{proposition}
    \begin{proof}
        We prove the lemma by induction. Let $1 \leq i \leq n$ be fixed. Since:
        \[
            C_{i,n-i+2} = C_{i,n-i+1} - D_{i,n-i+2} + N_{i,n-i+2},
        \]
        from Lemma \ref{binomiale_binomiale}, $C_{i,n-i+1} - D_{i,n-i+2} \mid \mathcal{F}_{n} \sim \mathcal{B}(C_{i,n-i+1}, 1-\delta_{i,n-i+1})$ and $N_{i,n-i+2} \mid \mathcal{F}_{n} \sim \mathcal{P}(\lambda_{n-i+2}E_{i})$.
Assume by induction that:
    \[C_{i,j} \mid \mathcal{F}_{n} \sim \mathcal{B}\left(C_{i,n-i+1}, \delta_{i,j}'\right) + \mathcal{P}\left(\lambda_{i,j}'E_{i}\right).
    \]
    Remark that
    \[
            C_{i,j+1} = C_{i,n-i+1} + (C_{i,j} - C_{i,n-i+1}) + N_{i,j+1} - D_{i,j+1}.
    \]
    and $D_{i,j+1} \sim \mathcal{B}(C_{i,j}, \delta_j) = \mathcal{B}(C_{i,n-i+1}, \delta_j) + \mathcal{B}(C_{i,j} - C_{i,n-i+1}, \delta_j)$, where the right-hand side should be interpreted as the sum of two independent random variables. \\ 
    Under \textit{H2'}, by Lemma \ref{binomiale_binomiale},
    \[
        C_{i,n-i+1} - \mathcal{B}(C_{i,n-i+1}, \delta_j) \mid \mathcal{F}_{n} \sim \mathcal{B}\left(C_{i,n-i+1}, \delta_{i,j+1}'\right),
    \]
    and
    \[
        (C_{i,j} - C_{i,n-i+1}) - \mathcal{B}(C_{i,j} - C_{i,n-i+1}, \delta_j) \mid \mathcal{F}_n \sim \mathcal{P}\left(\lambda_{i,j}'E_{i}(1-\delta_j\right))
    \]
    and finally, by Lemma \ref{poisson_binomiale}
    \[
         C_{i,j+1} \mid \mathcal{F}_{n} \sim \mathcal{B}\left(C_{i,n-i+1}, \delta_{i,j+1}'\right) + \mathcal{P}\left(\lambda_{i,j+1}'E_{i}\right).
    \]
    \end{proof}

The result indicates that, following the observation of $\mathcal{F}_{n}$, the distribution of an unobserved $C_{i,j}$ can be described as the sum of two components: current claims that exceed the priority threshold and are likely to stay above it, and new claims that may initially rise above the priority but might later fall below it.

If we are interested in estimating $\mathbb{E}(C_{i,n} \mid \mathcal{F}_n)$, we can use the estimator $\widehat{C}_{i,n} := \widehat{\lambda}_{i, n}'E_i + \widehat{\delta}_{i, n}'C_{i, n-i+1}$ where $(\widehat{\lambda}_{i, n}', \widehat{\delta}_{i, n}')$ are estimators for $(\lambda_{i, n}', \delta_{i, n}')$. The sequences $(\lambda_{k})$ and $(\delta_k)$ are estimated by $(\widehat{\lambda}_{k})$ and $(\widehat{\delta}_k)$, as defined in \eqref{main_estimators}.

\bigskip

In practice, the Poisson distribution is commonly used for modeling the number of claims due to its simplicity and the assumption of independent arrivals with non-random intensity across all claims. However, it is often observed that the empirical variance exceeds the empirical mean, suggesting that the claim data might not fully adhere to the assumptions of the Poisson distribution.

For instance, if the intensity parameter of claim arrivals for each policy follows a Gamma distribution, the total number of claims is known to follow a Negative binomial distribution. This distribution is favored because it remains straightforward to use, accommodates excess variance, and converges to the Poisson distribution when this variance diminishes.

In the following subsection, we propose using the Negative binomial distribution instead of the Poisson, demonstrating that, with an appropriate assumption, it can yield comparable results.

\subsection{The Negative binomial assumption}

In this section, we aim to establish a Negative binomial framework that yields results similar to those obtained from the previous Poisson framework.
Specifically, we define the Negative binomial distribution as follows: let $N\sim \mathcal{NB}(r, p)$ where $r > 0$ and $0 < p < 1$. The probability mass function of $N$ is given by
\begin{equation}\label{p_bn}
    \mathbb{P}(N = n) = \frac{\Gamma(r+n)}{n!\Gamma(r)}p^r(1-p)^n, \quad n \in \mathbb{N}.
\end{equation}
We then proceed by introducing a new assumption that serves as a replacement for \textit{H2'}. The change of this assumption has an impact only on the sequence of random variables $(N_{i,j})_{1 \leq i,j \leq n}$.

    \begin{hypothese}
	\leavevmode
        \begin{itemize}
            \item[H2"] For $1 \leq j \leq n$, there exists $r_{j} \geq 0$ and for $1 \leq j \leq n - 1$, there exists $0 \leq \delta_{j} \leq 1$ such that
	\[
		\begin{aligned}
		N_{i,j} \mid \mathcal{F}_{i + j - 2} &\sim \mathcal{NB}\left(r_{j}E_{i}, p_{j}\right) \text{ in which }p_{j+1} := \frac{p_{j}}{1-\delta_{j}(1-p_{j})}, \ \ &1 \leq i \leq n,\\
		D_{i,j+1} \mid \mathcal{F}_{i + j - 1} &\sim \mathcal{B}(C_{i, j}, \delta_{j}), \ \ &1 \leq i \leq n. 
		\end{aligned}
	\]
        \end{itemize}
    \end{hypothese}

Similarly, \textit{H2"} implies \textit{H2} and \textit{H3} from Section \ref{general}. The structure of the $p$'s may not initially appear clear or intuitive. However, the representation of these parameters will be clarified later. The free parameters are the $(r_{j})_{1 \leq j \leq n}$ (which replace the $\lambda$'s from \textit{H2'}), the $(\delta_{j})_{1 \leq j \leq n-1}$ and $p_1 \in ]0, 1[$. There is only one additional parameter, compared to \textit{H2'}. This extra parameter governs the additional variance due to the specific configuration of the family $(p_{j})_{1 \leq j \leq n}$.

\begin{remarque}
The $p$'s can be explicitly expressed in terms of $p_{1}$ and the $\delta$'s:
\[
    p_{j} = \frac{p_1}{1 - (1-p_1)\left[\sum_{k=1}^{j-1}\delta_{k}\prod_{\ell = 1}^{k-1}(1-\delta_{\ell})\right]}, \ 1 \leq j \leq n.
\]
\end{remarque}
The above remark can be verified through direct induction. Estimating the $r$'s, $\delta$'s, and $p_1$ cannot yield explicit maximum likelihood estimates. For $\widehat{\delta}$'s, we can use the same estimator as the one defined in the Poisson framework in (\ref{main_estimators}). When $p_1$ is known, to estimate the $r$'s, we can use a moment estimator based on the expected value and set:
    \begin{equation}\label{est_r}
        \widehat{r}_{j} := \widehat{\lambda}_j\frac{p_j}{1-p_j}, \ \ 1 \leq j \leq n.
    \end{equation}
Finally, the estimation of $p_1$ can be computed numerically using maximum likelihood methods:

\begin{equation}\label{eq_mle}
	\widehat{p}_{1} \in \argmax \sum_{j=1}^{n} \sum_{i=1}^{n-j+1} \log f(N_{i,j}, \widehat{r}_{j} E_i, p_j),
\end{equation}
in which $(n, r, p) \mapsto f(n, r, p)$ denotes the probability mass function of the Negative binomial distribution with parameter $(r, p)$, see \eqref{p_bn}.

\smallbreak

It remains to explain why we choose this form for the $p$'s. It is a consequence of the following result, in order to have a consistent form, as in the Poisson case. To establish the distribution of the $(C_{i, j})$, we begin with a classic lemma.

    \begin{lemme}\label{negative_binomiale}
        Let $N$ be a random variable with distribution $\mathcal{NB}(r, p)$ for $r > 0$ and $0 < p < 1$, and $D$ a random variable such that $D \mid N \sim \mathcal{B}(N, \delta)$ for $0 < \delta < 1$. Then
        
        \[
            N - D \sim \mathcal{NB}(r, p') \ \text{ with } \ p' := \frac{p}{1-\delta(1-p)}.
        \]
    \end{lemme}

\begin{proof}
Let $k \in \mathbb{N}$.
    \begin{align*}
        \mathbb{P}(N - D = k) &= \sum_{n \in \mathbb{N}}\mathbb{P}\left(\left\{D = n - k\right\} \mid \left\{N = n\right\}\right)\mathbb{P}\left(N = n\right) \\
        &= \sum_{n \geq k}\frac{n!}{k!(n-k)!}\delta^{n-k}(1-\delta)^{k}\frac{\Gamma(r+n)}{n!\Gamma(r)}p^r(1-p)^n \\
        &= \frac{(1-\delta)^k}{k!}p^r\sum_{n \geq 0}\frac{\delta^{n}}{n!}\frac{\Gamma(r+n+k)}{\Gamma(r)}(1-p)^{n+k} \\
        &= \frac{[(1-\delta)(1-p)]^k}{k!}p^r\sum_{n \geq 0}\frac{\Gamma(r+n+k)}{n!\Gamma(r)}[\delta(1-p)]^{n} \\
        &= \frac{\Gamma(r+k)}{k!\Gamma(r)}\left[\frac{(1-\delta)(1-p)}{1-\delta(1-p)}\right]^k\left[\frac{p}{1-\delta(1-p)}\right]^{r}.
        \end{align*}
The final line is obtained by noting that $\sum_{n \geq 0}\frac{\Gamma(r+k+n)}{n!\Gamma(r+k)}[1-\delta(1-p)]^{r+k}[\delta(1-p)]^{n} = 1$.
\end{proof}

             \begin{proposition}\label{loi_nb}
        For all $1 \leq i,j \leq n$, we have
        \[
            C_{i,j} \sim \mathcal{NB}\left(r_{j}'E_i, p_j\right),
        \]
        with
        \[
            r_{j}' := \sum_{k=1}^{j}r_k.
        \]
    \end{proposition}

        \begin{proof}
    We prove the lemma by induction. Let $1 \leq i \leq n$ be fixed. By construction $C_{i,1} = N_{i,1} \sim \mathcal{NB}(r_{1}E_{i}, p_1)$.
Assume by induction that $C_{i,j}$ follows a Negative binomial distribution with parameter $(r_{j}'E_i, p_j)$. Recall that 
    \[
            C_{i,j+1} = C_{i,j} + N_{i,j+1} - D_{i,j+1}.
    \]
    Under \textit{H2''}, by Lemma \ref{negative_binomiale},
    \[
         C_{i,j} - D_{i,j+1} \sim \mathcal{NB}\left(r_{j}'E_{i}, p_{j+1}\right),
    \]
    and finally:
    \[
        C_{i,j+1} \sim \mathcal{NB}\left(r_{j+1}'E_{i}, p_{j+1}\right).
    \]
    \end{proof}

The form of the $p$'s can now be understood. Assuming new claims above the priority threshold follow a Negative binomial distribution and that some claims may later fall below the priority, we aim to maintain consistency at any point in time with the Negative binomial distribution. Consequently, this leads to the specific form of the $p$'s.

\smallbreak

The $p$'s are increasing, leading to a smaller excess of variance over time. Notably, the likelihood of claims dropping below the priority with probabilities $\delta$'s influences the Negative binomial distribution, including claims not yet reported in future development years. And more likely the claim are droping below the thereshold, faster decreases the excess of variance for the new claims.

\smallbreak

In particular, for the extremal cases, if $\delta_{j} = 0$, meaning that no claims drops below the priority, the excess of variance does not reduce. Conversely, when $\delta_{j} \rightarrow 1$, meaning all claims drop below the priority, the excess variance vanishes.

\smallbreak

The preceding result provides the distribution of $C_{n+1, n}$ given the related exposure. Additionally, we may be interested in the distribution of $C_{i, j} \mid \mathcal{F}_n$ for $i + j > n + 1$, as described in the following proposition.

\begin{proposition}
    \label{loi_fn2}
        For all $1 \leq i,j \leq n$ such that $i+j > n + 1$, we have
        \[
            C_{i,j} \mid \mathcal{F}_{n} \sim \mathcal{B}(C_{i,n-i+1}, \delta_{i,j}') + \mathcal{NB}\left(r_{j}'E_i, p_j\right),
        \]
        where the right-hand side should be interpreted as the sum of two independent random variables, and with
        \begin{align*}
            \delta_{i,j}' &:= \prod_{k = n-i+1}^{j-1}(1-\delta_{k}), \\
            r_{j}' &:= \sum_{k = n-i+2}^{j}r_k.
        \end{align*}
\end{proposition}

\begin{proof}
The proof follows the same reasoning as in Proposition \ref{loi_fn}, with the key difference being the application of Lemma \ref{negative_binomiale} in place of Lemma \ref{poisson_binomiale}.
\end{proof}

If we are interested in estimating $\mathbb{E}(C_{i,n} \mid \mathcal{F}_n)$, using \eqref{est_r} for the estimators of the $r$'s leads to the same estimator as the one suggested for Proposition \ref{loi_fn}.

\section{Bootstrap methodology}

In \cite{liu2009bootstrap}, the author discusses a bootstrap methodology for the general Schnieper model to resimulate the $\lambda$'s and $\delta$'s, accounting for uncertainty in the parameters. They also simulate a Gaussian random variable to integrate the internal randomness of the process for each development stage. This follows the main ideas of the non-parametric bootstrap as summarized in \cite{england2006predictive}.

\medbreak

Here, we present a distinct approach that utilizes the specific framework of claim numbers and proposes a comprehensive parametric bootstrap methodology without computing residuals or making any additional assumption.

\subsection{The Poisson case}

Let $M \in \mathbb{N}^*$ be the total number of bootstrap simulations to be performed. To account for the inherent randomness of the $\lambda$'s and $\delta$'s, for $1 \leq m \leq M$. To do it efficiently, we shall use the following lemma.

\begin{lemme}\label{bootstrap_param_poisson}
        Under H1 and H2', we have:
        \[
	\begin{aligned}
            \widehat{\lambda}_{j} \mid \mathcal{B}_{j-1} &\sim \frac{\mathcal{P}\left(\lambda_{j}\sum_{i=1}^{n-j+1}E_{i}\right)}{\sum_{i=1}^{n-j+1}E_{i}}, \ \ &1 \leq j \leq n, \\ && \\
            \widehat{\delta}_{j} \mid \mathcal{B}_{j} &\sim \frac{\mathcal{B}\left(\sum_{i=1}^{n-j}C_{i,j}, \delta_{j}\right)}{\sum_{i=1}^{n-j}C_{i,j}}, \ \ &1 \leq j \leq n-1,
        \end{aligned}
	\]
where $\mathcal{B}_{k} := \sigma(N_{i, j}, D_{i, j} \mid i+j \leq n+1, j \leq k)$.
\end{lemme}

    \begin{proof}
    Direct consequence of \textit{H1} and \textit{H2'}.
    \end{proof}

This provides a direct method to simulate the bootstrapped $\lambda$'s and $\delta$'s.
\[
	\begin{aligned}
            (\widehat{\lambda}_{j}^m)_{1 \leq m \leq M}  &\overset{i.i.d.}{\sim} \frac{\mathcal{P}\left(\widehat{\lambda}_{j}\sum_{i=1}^{n-j+1}E_{i}\right)}{\sum_{i=1}^{n-j+1}E_{i}}, \ \ &1 \leq j \leq n, \\ && \\
            (\widehat{\delta}_{j}^m)_{1 \leq m \leq M}  &\overset{i.i.d.}{\sim} \frac{\mathcal{B}\left(\sum_{i=1}^{n-j}C_{i,j}, \widehat{\delta}_{j}\right)}{\sum_{i=1}^{n-j}C_{i,j}}, \ \ &1 \leq j \leq n-1,
        \end{aligned}    
\]
where $(\widehat{\lambda}_j)_{1 \leq j \leq n}$ and $(\widehat{\delta}_j)_{1 \leq j \leq n-1}$ come from $\eqref{main_estimators}$ and the $C$'s are the observed data.

\paragraph{Bootstrap simulation of $C_{n+1, n}$.} \ Following Proposition \ref{loi_poisson}, the bootstrap simulation is performed as follows:
    \[
        C_{n+1, n}^{m} \sim \mathcal{P}(\widehat{\lambda}_{n}'^{m}E_{n+1}), \ \ 1 \leq m \leq M.
    \]

\paragraph{Bootstrap simulation of $C_{i, j} \mid \mathcal{F}_{n}$.} \ For the lower triangle, the simulation is conducted using:
		\[
            \begin{aligned}
            C_{i,n-i+1}^{m} &:= C_{i,n-i+1}, \ \ &1 \leq i \leq n, \\
            C_{i,j+1}^{m} &:= C_{i,j}^{m} + \mathcal{P}\left(\widehat{\lambda}_{j+1}^{m}E_{i}\right)  - \mathcal{B}\left(C_{i,j}^{m}, \widehat{\delta}_{j}^{m}\right), \ \ &n-i+1 \leq j \leq n-1.
        \end{aligned}
		\]
On the right-hand side of the last equality, the difference of the two distributions should be interpreted as the difference of two independent random variables. This procedure generates a bootstrap distribution for the random variable $C_{i,j} \mid \mathcal{F}_n$ on the lower triangle. In this process, the uncertainty associated with the estimators of the parameters is integrated.

\begin{remarque}\label{bootstrap_cin}
Based on Proposition \ref{loi_fn}, when our focus is on $C_{i,n} \mid \mathcal{F}_{n}$, we can efficiently simulate $C_{i,n}^{m}$, using
        \[
            C_{i,n}^{m} \sim \mathcal{B}(C_{i,n-i+1}, \widehat{\delta}_{i,n}'^m) + \mathcal{P}\left(\widehat{\lambda}_{i,n}'^mE_{i}\right),
        \]
where the right-hand side should be interpreted as the sum of two independent random variables, and with
        \begin{align*}
            \widehat{\delta}_{i,n}'^m &:= \prod_{k = n-i+1}^{n-1}(1-\widehat{\delta}_{k}^{m}), \\
            \widehat{\lambda}_{i,n}'^{m} &:= \sum_{k = n-i+2}^{n}\widehat{\lambda}_{k}^{m}\left(\prod_{\ell = k}^{n-1}(1-\widehat{\delta}_{\ell}^{m})\right).
        \end{align*}
This provides a more efficient algorithm.
\end{remarque}

\subsection{The Negative binomial case}

We extend the Poisson model to fit the Negative binomial framework. To account for the variability in the parameters $r$'s, $\delta$'s and $p_1$, for $1 \leq m \leq M$. We have the following lemma, similar to Lemma \ref{bootstrap_param_poisson}.

\begin{lemme}
        Under H1 and H2", we have:
        \[
	\begin{aligned}
            \widehat{\lambda}_{j} \mid \mathcal{B}_{j-1} &\sim \frac{\mathcal{NB}\left(r_{j}\sum_{i=1}^{n-j+1}E_{i}, p_j\right)}{\sum_{i=1}^{n-j+1}E_{i}}, \ \ &1 \leq j \leq n, \\ && \\
            \widehat{\delta}_{j} \mid \mathcal{B}_{j} &\sim \frac{\mathcal{B}\left(\sum_{i=1}^{n-j}C_{i,j}, \delta_{j}\right)}{\sum_{i=1}^{n-j}C_{i,j}}, \ \ &1 \leq j \leq n-1,
        \end{aligned}
	\]
where $\mathcal{B}_{k} := \sigma(N_{i, j}, D_{i, j} \mid i+j \leq n+1, j \leq k)$.
\end{lemme}

    \begin{proof}
    Direct consequence of \textit{H1} and \textit{H2"}.
    \end{proof}

However, unlike the Poisson case, we cannot apply the above lemma straightforwardly since it does not provide the distribution of $p_1$. Additionally, the estimator of $p_1$, defined in \eqref{eq_mle}, is non-trivially dependent on the $N_{i,j}$.

Nonetheless, for the $\delta's$, we can proceed as follows:
\[
	\begin{aligned}
            (\widehat{\delta}_{j}^m)_{1 \leq m \leq M}  &\overset{i.i.d.}{\sim} \frac{\mathcal{B}\left(\sum_{i=1}^{n-j}C_{i,j}, \widehat{\delta}_{j}\right)}{\sum_{i=1}^{n-j}C_{i,j}}, \quad &1 \leq j \leq n-1.
        \end{aligned}    
\]
For the $r$'s and $p_1$, we resimulate the upper triangle $N_{i,j}$.
	\[
        \begin{aligned}
            N_{i,j}^{m} &\sim \mathcal{NB}\left(\widehat{r}_{j}E_{i}, \widehat{p}_j\right), \quad &1 \leq i \leq n, \quad 1 \leq j \leq n-i.
\end{aligned} 
        \]

From these upper triangles, we can estimate the $r$'s and $p_1$. These are denoted respectively as $(\widehat{r}_{j}^{m})_{1 \leq j \leq n}$ and $\widehat{p}_1^m$.

\paragraph{Bootstrap simulation of $C_{n+1, n}$.} \ Following Proposition \ref{loi_nb}, the simulation is straightforward:
    \[
        C_{n+1, n}^{m} \sim \mathcal{NB}(\widehat{r}_{n}^{'m}E_{n+1}, \widehat{p}_n^m), \ \ 1 \leq m \leq M.
    \]

\paragraph{Bootstrap simulation of $C_{i, j} \mid \mathcal{F}_{n}$.} \ For the lower triangle, we simulate :
	\[
            \begin{aligned}
            C_{i,n-i+1}^{m} &:= C_{i,n-i+1}, \ \ &1 \leq i \leq n, \\
            C_{i,j+1}^{m} &:= C_{i,j}^{m} + \mathcal{NB}\left(\widehat{r}_{j+1}^{m}E_{i}, \widehat{p}_{j+1}^m\right)  - \mathcal{B}\left(C_{i,j}^{m}, \widehat{\delta}_{j}^{m}\right), \ \ &n-i+1 \leq j \leq n-1.
        \end{aligned}
	\]

On the right-hand side of the last equality, the difference of the two distributions should be interpreted as the difference of two independent random variables. Similarly to the approach discussed in Remark \ref{bootstrap_cin} for the Poisson case, if our focus is solely on the distribution $C_{i, n} \mid \mathcal{F}_{n}$, we can bypass simulating the entire lower triangle using Proposition \ref{loi_fn2}.

\section{Example}

We present two triangles of simulated data to illustrate both cases with $n = 6$ years of observations. For the first one, the exposure and the $C$ triangle are:
        
        \begin{table}[H]
        \begin{center}\footnotesize
        \begin{tabular}{|c|c|c c c c c c|}
  \hline
  $i$ & $E_{i}$ & $C_{i,1}$ & $C_{i,2}$ & $C_{i,3}$ & $C_{i,4}$ & $C_{i,5}$ & $C_{i,6}$ \\
  \hline
  1 & 20 & 5 & 9 & 11 & 12 & 13 & 11 \\
  2 & 25 & 11 & 16 & 13 & 11 & 17 &\\
  3 & 32 & 9 & 17 & 22 & 22 & & \\
  4 & 38 & 10 & 10 & 11 & & & \\
  5 & 42 & 17 & 18 & & & & \\
  6 & 45 & 14 & & & & & \\
  \hline
\end{tabular}
\end{center}
\end{table}
whose decomposition in $N$ and $D$ is:
        \begin{table}[H]
        \begin{center}\footnotesize
        \begin{tabular}{|c|c c c c c c|}
  \hline
  $i$ & $N_{i,1}$ & $N_{i,2}$ & $N_{i,3}$ & $N_{i,4}$ & $N_{i,5}$ & $N_{i,6}$ \\
  \hline
  1 & 5 & 4 & 5 & 2 & 1 & 0 \\
  2 & 11 & 9 & 4 & 4 & 6 &\\
  3 & 9 & 14 & 9 & 3 & & \\
  4 & 10 & 7 & 5 & & & \\
  5 & 17 & 10 & & & & \\
  6 & 14 & & & & & \\
  \hline
\end{tabular} \ \ \ \begin{tabular}{|c|c c c c c c|}
  \hline
  $i$ & $D_{i,1}$ & $D_{i,2}$ & $D_{i,3}$ & $D_{i,4}$ & $D_{i,5}$ & $D_{i,6}$ \\
  \hline
  1 & 0 & 0 & 3 & 1 & 0 & 2 \\
  2 & 0 & 4 & 7 & 6 & 0 & \\
  3 & 0 & 6 & 4 & 3 & & \\
  4 & 0 & 7 & 4 & & & \\
  5 & 0 & 9 & & & & \\
  6 & 0 & & & & & \\
  \hline
\end{tabular}
\end{center}
\end{table}
We can derive directly the $\lambda$'s and $\delta$'s.
      \begin{table}[H]
        \begin{center}\footnotesize
        \begin{tabular}{|c|c c c c c c|}
  \hline
    $j$ & 1 & 2 & 3 & 4 & 5 & 6\\
    \hline
    $\widehat{\lambda}_{j}$ & 0.327 & 0.28 & 0.2 & 0.117 & 0.156 & 0  \\
    \hline
    $\widehat{\delta}_{j}$ & 0.5 & 0.346 & 0.217 & 0 & 0.154 & \textendash \\
  \hline
\end{tabular}
\end{center}
\end{table}

Let $E_{n+1} = 50$ be the exposure for the upcoming year. Under the Poisson assumption, using the estimator of the intensity leads to:
\[
    C_{n+1,n} \sim \mathcal{P}\left(27.752\right)
\]
Under the Negative binomial assumption and utilizing the $\lambda$'s and the $\delta$'s, and computing $p_1$ by maximum likelihood leads to optimal $p_1 \rightarrow 1$. In this case, the Negative binomial distribution converges to the Poisson distribution: the assumption does not appear suitable.

\medbreak

To account for the uncertainty in the unknown parameters, we can use the bootstrap procedure. We get that the variance of $C_{n+1, n}$ is now around 53.361, which is notably higher than the variance obtained when using the Poisson distribution with the estimated parameter. Figure \ref{bootstrap_hist} illustrates the histogram of $C_{n+1, n}$ with the distribution $\mathcal{P}\left(27.752\right)$ (lighter on the left) compared to the distribution obtained from the bootstrap procedure (darker on the right).

\begin{figure}[H]
\centering
\includegraphics[scale=0.6]{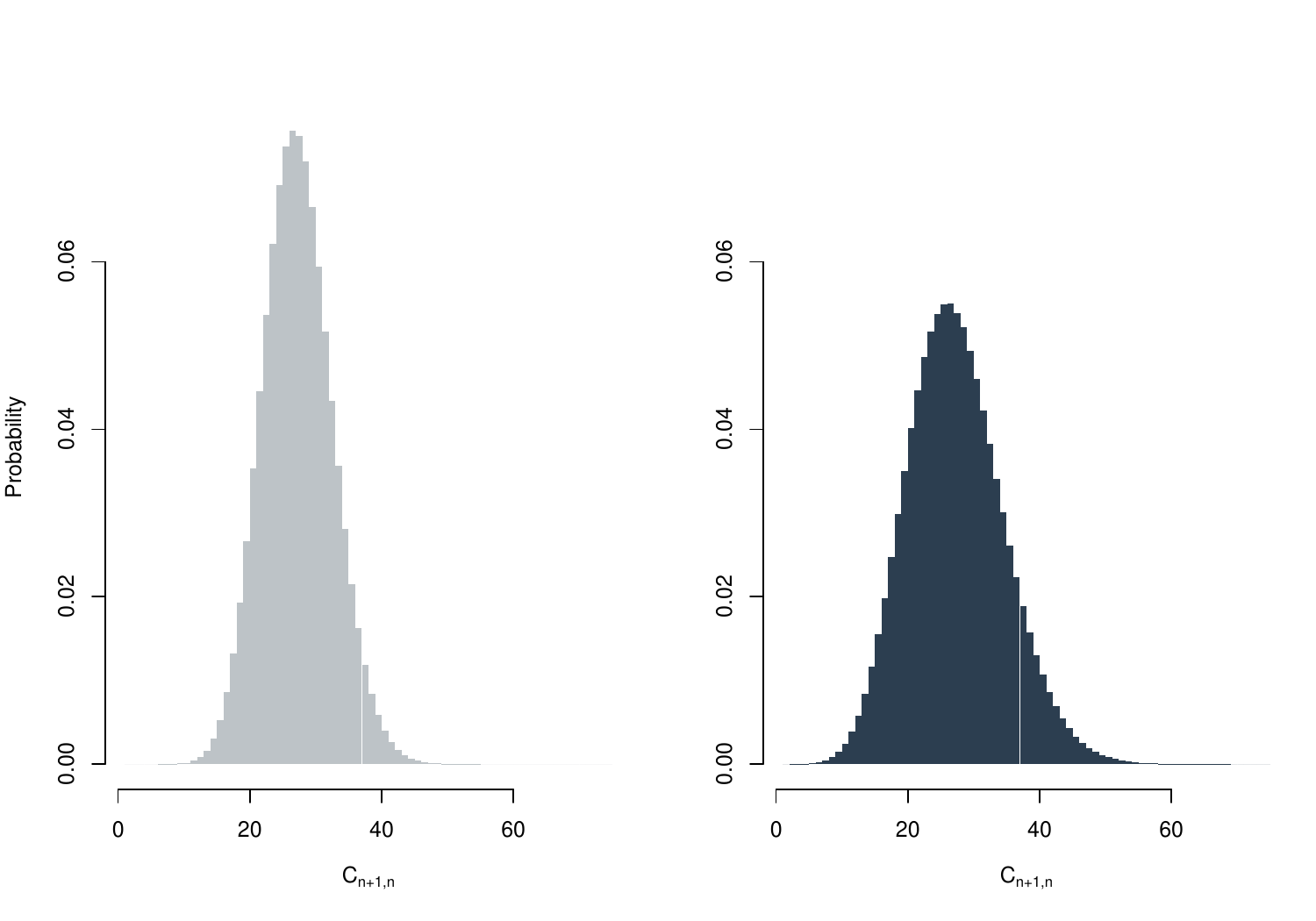}
\caption{Comparison of the distribution of the estimated distribution $\mathcal{P}\left(27.752\right)$ (on the left), and the associated bootstrap distribution (on the right); with $M = 10^7$ simulations.}\label{bootstrap_hist}
\end{figure}

\bigskip

We now introduce a second set of triangles, the exposure and the $C$ triangle are:
        
        \begin{table}[H]
        \begin{center}\footnotesize
        \begin{tabular}{|c|c|c c c c c c|}
  \hline
  $i$ & $E_{i}$ & $C_{i,1}$ & $C_{i,2}$ & $C_{i,3}$ & $C_{i,4}$ & $C_{i,5}$ & $C_{i,6}$ \\
  \hline
  1 & 20 & 8 & 4 & 12 & 12 & 14 & 13 \\
  2 & 25 & 3 & 5 & 7 & 10 & 15 &\\
  3 & 32 & 5 & 10 & 11 & 9 & & \\
  4 & 38 & 27 & 20 & 29 & & & \\
  5 & 42 & 23 & 18 & & & & \\
  6 & 45 & 14 & & & & & \\
  \hline
\end{tabular}
\end{center}
\end{table}

whose decomposition in $N$ and $D$ is:
        \begin{table}[H]
        \begin{center}\footnotesize
        \begin{tabular}{|c|c c c c c c|}
  \hline
  $i$ & $N_{i,1}$ & $N_{i,2}$ & $N_{i,3}$ & $N_{i,4}$ & $N_{i,5}$ & $N_{i,6}$ \\
  \hline
  1 & 8 & 3 & 9 & 4 & 3 & 0 \\
  2 & 3 & 5 & 4 & 3 & 6 &\\
  3 & 5 & 7 & 3 & 3 & & \\
  4 & 27 & 8 & 13 & & & \\
  5 & 23 & 7 & & & & \\
  6 & 14 & & & & & \\
  \hline
\end{tabular} \ \ \ \begin{tabular}{|c|c c c c c c|}
  \hline
  $i$ & $D_{i,1}$ & $D_{i,2}$ & $D_{i,3}$ & $D_{i,4}$ & $D_{i,5}$ & $D_{i,6}$ \\
  \hline
  1 & 0 & 7 & 1 & 4 & 1 & 1 \\
  2 & 0 & 3 & 2 & 0 & 1 & \\
  3 & 0 & 2 & 2 & 5 & & \\
  4 & 0 & 15 & 4 & & & \\
  5 & 0 & 12 & & & & \\
  6 & 0 & & & & & \\
  \hline
\end{tabular}
\end{center}
\end{table}

\FloatBarrier

Again, we can derive directly the $\lambda$'s and $\delta$'s.

      \begin{table}[H]
        \begin{center}\footnotesize
        \begin{tabular}{|c|c c c c c c|}
  \hline
    $j$ & 1 & 2 & 3 & 4 & 5 & 6\\
    \hline
    $\widehat{\lambda}_{j}$ & 0.396 & 0.191 & 0.252 & 0.13 & 0.2 & 0  \\
    \hline
    $\widehat{\delta}_{j}$ & 0.591 & 0.231 & 0.3 & 0.091 & 0.071 & \textendash \\
  \hline
\end{tabular}
\end{center}
\end{table}

Given an exposure of $E_{n+1} = 50$ for the next year and assuming a Poisson distribution, we get that
\[
    C_{n+1,n} \sim \mathcal{P}\left(30.243\right)
\]
Computing $p_1$ by maximum likelihood does not lead to $p_1 \rightarrow 1$ anymore. We find $\widehat{p}_1 = 0.397$. This implies that

      \begin{table}[H]
        \begin{center}\footnotesize
        \begin{tabular}{|c|c c c c c c|}
  \hline
    $j$ & 1 & 2 & 3 & 4 & 5 & 6\\
    \hline
    $\widehat{p}_{j}$ & 0.397 & 0.616 & 0.676 & 0.749 & 0.767 & 0.78  \\
    \hline
    $\widehat{r}_{j}$ & 0.263 & 0.402 & 0.418 & 0.448 & 0.328 & 0.177 \\
  \hline
\end{tabular}
\end{center}
\end{table}

In particular, under the Negative binomial assumption,
\[
    C_{n+1,n} \sim \mathcal{NB}\left(106.94, 0.780\right).
\]
By construction, the expected value of $C_{n+1,n}$ remains at 30.243, but the variance increases to 38.796.

To choose between the two assumptions, we can calculate the log-likelihood and AIC for both cases. Table \ref{aic} presents the results.

      \begin{table}[H]
        \begin{center}\footnotesize
        \begin{tabular}{|c|c c|}
  \hline
     & log-L. & AIC\\
    \hline
    $\mathcal{P}$ & -53.937 & 119.875  \\
    \hline
    $\mathcal{NB}$ & -50.793 & 115.586  \\
  \hline
\end{tabular}
\end{center}
\caption{Log-likelihood and AIC for both assumptions.}
\label{aic}
\end{table}

It appears that using the Negative binomial distribution is the most suitable choice in this scenario. Given that the $\delta$'s are also part of the definition of the $p$'s, one might question whether it would be beneficial to estimate both $(p_{1}, (\delta_{j})_{1 \leq j \leq n-1})$ simultaneously using data from both $N$ and $D$. Let $(\widetilde{p}_1, (\widetilde{\delta}_{j})_{1 \leq j \leq n-1})$ denote the new estimators. Table \ref{pd_comparison} provides a comparison, which shows that the difference is minimal.

      \begin{table}[H]
        \begin{center}\footnotesize
        \begin{tabular}{|c|c c c c c|}
  \hline
    $j$ & 1 & 2 & 3 & 4 & 5\\
    \hline
    $\widehat{\delta}_{j}$ & 0.591 & 0.231 & 0.3 & 0.091 & 0.071  \\
    \hline
    $\widetilde{\delta}_{j}$ & 0.601 & 0.232 & 0.305 & 0.092 & 0.071 \\
  \hline
\end{tabular} \\ \bigskip
        \begin{tabular}{|c|c c c c c c|}
  \hline
    $j$ & 1 & 2 & 3 & 4 & 5 & 6\\
    \hline
    $\widehat{p}_{j}$ & 0.397 & 0.616 & 0.676 & 0.749 & 0.767 & 0.78  \\
    \hline
    $\widetilde{p}_{j}$ & 0.393 & 0.619 & 0.679 & 0.753 & 0.77 & 0.783  \\
  \hline
\end{tabular}
\end{center}
\caption{Comparison of the two estimation methods.}
\label{pd_comparison}
\end{table}

Figure \ref{graphboth} shows the distribution from the Poisson assumption (lighter on the left) compared with the distribution from the negative binomial assumption (darker on the right).

\begin{figure}[H]
\centering
\includegraphics[scale=0.6]{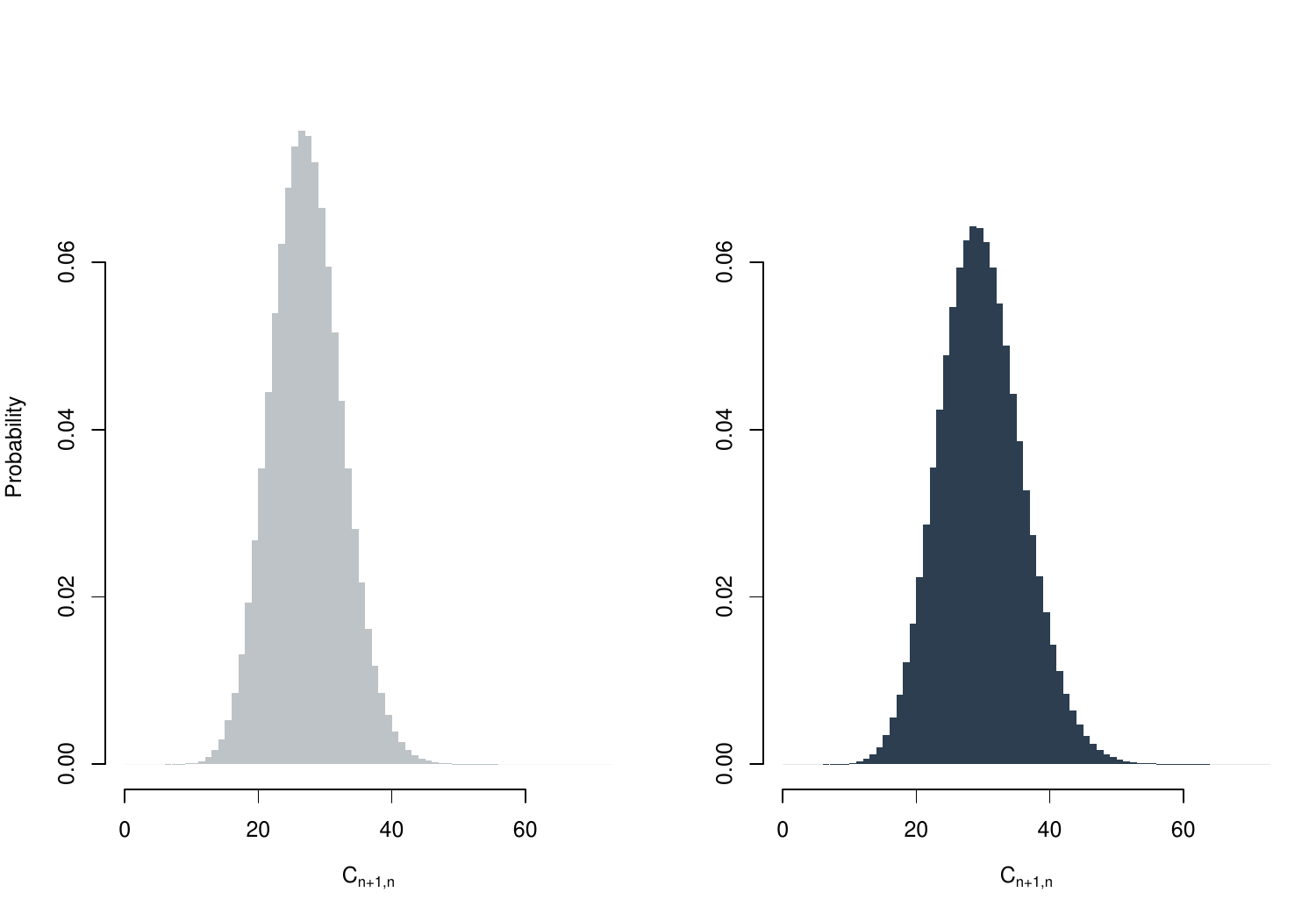}
\caption{Comparison of the distribution of the estimated distribution $\mathcal{P}\left(30.243\right)$ (on the left), and the distribution of the estimated distribution $\mathcal{NB}\left(106.939, 0.780\right)$ (on the right).}\label{graphboth}
\end{figure}

Figure \ref{graphboth2} illustrates the bootstrap distribution from the Poisson assumption (lighter on the left) and from the Binomial negative assumption (darker on the right). In each simulation using the negative binomial approach, if the estimated $\widehat{p}_{1}$ was close to 1, suggesting that the Poisson distribution was a better fit, the simulation was conducted using the Poisson framework. The bootstrap results show that the variance of the distribution under the Poisson assumption is 62.633, while the variance under the negative binomial assumption is 67.658.

\begin{figure}[H]
\centering
\includegraphics[scale=0.6]{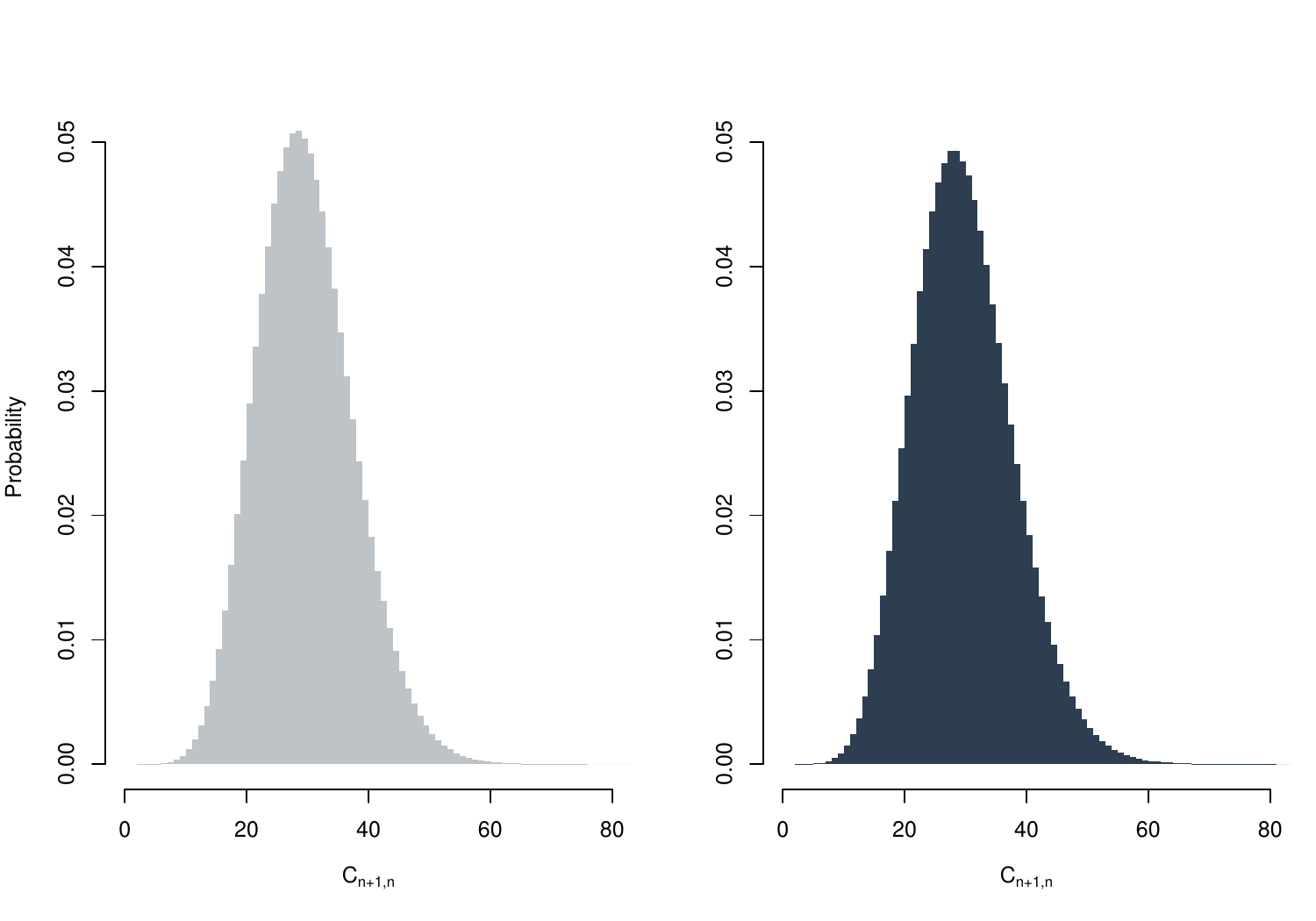}
\caption{Comparison of the Bootstrap distribution obtained with the Poisson assumption (on the left), and with the Negative binomial assumption (on the right); with $M = 10^7$ simulations.}\label{graphboth2}
\end{figure}

\FloatBarrier

\section*{Acknowledgments}

The author acknowledges the financial support provided by the \emph{Fondation Natixis}.

\bibliographystyle{plain}
\bibliography{bibliographie}

\end{document}